\documentclass[letterpaper, 10 pt, conference]{ieeeconf}
\IEEEoverridecommandlockouts

\usepackage{cite}
\usepackage{amsmath}
\usepackage{amssymb}
\usepackage{amsfonts}
\usepackage[mathscr]{eucal}
\usepackage{algorithm}
\usepackage{algpseudocode}
\usepackage{mathtools}
\usepackage{graphicx}
\usepackage{textcomp}
\let\labelindent\relax
\usepackage{enumitem}
\usepackage{xcolor}
\usepackage{url}
\usepackage[normalem]{ulem}
\usepackage{tikz}
\usetikzlibrary{calc, decorations.pathmorphing, arrows.meta}
\usepackage{cleveref}
\usepackage{booktabs} 
\graphicspath{{figures/}}
\newcommand{\R}{\mathbb{R}}

\newtheorem{theorem}{Theorem}

\newtheorem{assumption}{Assumption}

\begin{document}
\title{Direct and Indirect Data-Driven Control\\ with Prior Information about the Equilibrium Manifold}

\author{Yi-Chun Liao$^{\,a}$, Valentina Breschi$^{\,b}$, and Marco M. Nicotra$^{\,a}$%
\thanks{$^a$ Department of Electrical, Computer, and Energy Engineering, University of Colorado,  Boulder, CO 80309, USA. This research is supported by the NSF-CMMI Award \#2046212.} 
\thanks{$^b$ Department of Electrical Engineering, Eindhoven University of Technology, Eindhoven, 5612AZ, The Netherlands.}}

\maketitle

\begin{abstract}
By hinging on the assumption that a system to be controlled is fully unknown, many data-driven control approaches do not leverage available or readily inferable priors. In contrast to this viewpoint, this paper analyzes the impact of using the system's equilibrium subspace to inform direct and indirect linear quadratic regulation. For the indirect case, we show how including a constraint on the equilibrium subspace in the identification problem changes the statistical properties of the learned model. In particular, we show that enforcing consistency with respect to the equilibrium subspace leads to a reduction in the estimator variance that, in turn, enhances model-based control performance. In the direct case, we show how this prior can be leveraged to gain insight into the controlled system without requiring an explicit identification step. These results are supported by both numerical and experimental evidence, showcasing the advantages of explicitly leveraging the equilibrium manifold as a prior in data-driven control.         
\end{abstract}

\section{Introduction}
Data-driven control offers a systematic framework for synthesizing controllers in the absence of a model for the controlled system \cite{DD_Control,DD_Survey}. In particular, data-driven methods can be divided into indirect and direct approaches. The former uses data to identify a model of the controlled system, and then computes the feedback gain by relying on such a model. The latter use the data for the synthesis of the control gains, without an explicit identification step. 

The majority of data-driven control literature focuses on systems that are fully unknown. Yet, there are many applications where the system properties are partially known. In these cases, the performance of data-driven controllers might be improved by incorporating prior knowledge about the system into the identification or direct design of the controller. In the context of indirect methods, existing approaches allow to incorporate priors on asymptotically stable subspaces \cite{van_der_Veen,Umenberge}, known eigenvalues \cite{Miller}, frequency gains \cite{Khosravi_frequency_gain}, steady-state gains \cite{Khosravi}, or structural properties such as positivity \cite{Alberto, Khosravi_positive}, passivity \cite{Shali}, and stabilizability \cite{Shakouri}. As for direct methods, only the incorporation of prior information on system parameters, either fully-known \cite{Huang} or within given bounds \cite{Berberich}, has been explored.

This paper investigates the advantages of incorporating prior knowledge about the equilibrium subspace of the system, which is often available a priori or can readily be estimated from steady-state measurements. For indirect methods, coherence with respect to known equilibrium manifolds is achieved by defining a Linear Matrix Equality that enforces structural consistency between the open loop model and the equilibrium subspace. We also show how such a constraint shapes the properties of  the identified system. For direct methods, we show that prior information about the equilibrium manifold cannot be meaningfully embedded into the direct control design problem. Nonetheless, we show how the knowledge about the equilibrium manifold can be used to recover a model of the controlled system from the designed control law. Numerical and experimental validations showcase the advantages of considering the information on equilibrium manifold in both direct and indirect methods.

\paragraph*{Notation} The set of natural numbers including zero is denoted as $\mathbb{N}_{0}$, while the set of real numbers is denoted by $\mathbb{R}$. The sets of real column vectors of dimension $m$ and real matrices of dimension $m\times n$ are respectively denoted as $\mathbb{R}^{m}$ and $\mathbb{R}^{m \times n}$. Given a matrix $M \in \mathbb{R}^{m \times n}$, we define its Frobenius norm as $\|M\|_F^2=\mbox{trace}(MM^\top)$, its vectorization as $\mbox{vec}(M)$, and we compactly indicate that it is positive definite as $M \succ 0$. If $M$ is full-row rank, its Moore-Penrose pseudoinverse is $M^\dagger=M^\top(MM^\top)^{-1}$. Given a random vector $a$, we indicate its expected value as $\mathbb{E}[a]$ and its covariance $\mathbb{E}[(a-\mathbb{E}[a])(a-\mathbb{E}[a])^{\top}]$ as $\mbox{cov}(a)$.

\section{Setting \& goal}\label{sec:setting}
Let us consider the following linear, time-invariant system
\begin{equation}\label{eq:open-loop}
    x_{k+1}=A x_k+ B u_k+w_k,~~~k \in \mathbb{N}_{0},
\end{equation}
where $x\in\mathbb R^n$ is the system's state, $u\in\mathbb R^m$ is the control input, and $w_k\in\R^n$ is the process noise corrupting the state evolution. Assume that $(A,B)$ is stabilizable and 
let the process noise satisfy the following.
\begin{assumption}[Process noise]\label{assump:process_noise}
    The process noise $w$ is zero-mean, white, and with finite covariance and uncorrelated with the initial state of the system and the input, namely\vspace{-6pt}
\begin{subequations}\label{eq:noise_properties}
    \begin{align}
        &\mathbb{E}[w_k]=0,~~\forall k \in \mathbb{N}_{0},\\
        &\Sigma_w=\mathbb{E}[w_kw_k^\top]=\sigma_w^2I<\infty,~~\forall k \in \mathbb{N}_{0},\\
        &\mathbb{E}[w_kx_0^\top]=0,~~~~\mathbb{E}[w_ku_\tau^\top]=0,~~\forall k,\tau \in \mathbb{N}_{0}.\vspace{6pt}
    \end{align}
    \end{subequations}
\end{assumption}

Suppose that the matrices $A \in \mathbb{R}^{n\times n}$ and $B \in \mathbb{R}^{n \times m}$ are \emph{unknown}, but that we have access to a state/input dataset
\begin{equation}\label{eq:data}
    \mathcal{D}_T=\{u_k,x_k\}_{k=0}^{T},
\end{equation}
collected by feeding the system with a persistently exciting input sequence, i.e., $\{u_k\}_{k=0}^{T}$ guarantees 
\begin{equation}\label{eq:PE_condition}
    \mbox{rank}\left(\begin{bmatrix}
        X_0^\top & U_0^\top
    \end{bmatrix}^\top\right)=n+m,
\end{equation}
where
\begin{equation}\label{eq:data_matrices0}
    \begin{aligned}
        & X_0 := \begin{bmatrix}
            x_0 & x_1 & \cdots & x_{T-1} 
        \end{bmatrix},\\
        &U_0 := \begin{bmatrix}
            u_0 & u_1 & \cdots & u_{T-1} 
        \end{bmatrix}.
    \end{aligned}
\end{equation}
Under these assumptions, we aim to design a Linear Quadratic Regulator (LQR) for the \emph{unknown} system, i.e., a static state feedback law
\begin{equation}\label{eq:LQR}
    u_k=Kx_k,~~~k \in \mathbb{N}_0,
\end{equation}
by levering \emph{priors} on the \emph{equilibrium manifold} of \eqref{eq:open-loop}. 

\section{Preliminaries}\label{sec:preliminaries}
Before analyzing the impact of leveraging priors on equilibrium subspaces into indirect and direct linear quadratic control approaches, we briefly introduce equilibrium subspaces and linear quadratic regulation.   
\subsection{Equilibrium Subspace}
Consider the linear discrete-time system in \eqref{eq:open-loop}. If $B$ is full rank, there exist a full-column rank matrix $G_x\in\R^{n\times m}$ and a matrix $G_u\in\R^{m\times m}$ such that
\begin{equation}\label{eq:nullSpace}
    \begin{bmatrix}
    G_x \\G_u       
    \end{bmatrix}\in\mbox{null}\left(\begin{bmatrix}
        A-I&B
    \end{bmatrix}\right).
\end{equation}
This relation trivially entails the following:
\begin{equation}\label{eq:G_xr_condition}
    G_xr = AG_xr + BG_ur,\qquad \forall r\in\R^m.
\end{equation}
Therefore, the matrices $(G_x,G_u)$ provide a mapping between $r\in\R^m$ and the steady-state equilibria of \eqref{eq:open-loop}. This condition can be further rewritten as
\begin{equation}\label{eq:Equilibrium_subspace_general_condition}
    \begin{bmatrix}
        A & B
    \end{bmatrix}\begin{bmatrix}
        G_x\\
        G_u
    \end{bmatrix}=G_x,
\end{equation}
leading to an equality constraint that can be enforced within identification routines whenever information on $G_x$ and $G_u$ are available a priori. Given a static state-feedback gain $K$ such that $A\!+\!BK$ is Schur, the control law
\begin{equation}\label{eq:feedbackLaw}
    u_k=G_ur+K(x_k-G_xr)
\end{equation}
guarantees that $\bar x =G_xr$ is an Input-to-State Stable (ISS) equilibrium point\cite[Sec.~3.5, App.~B.6]{rawlings2017} for the closed-loop system 
\begin{equation}\label{eq:closed-loop}
    x_{k+1}\!=\!\!\underbrace{(A\!+\!BK)}_{:=A_\pi}x_k\!+\!\underbrace{B(G_u\!\!-\!KG_x)}_{:=B_\pi}r+w_k,~~~\forall k \in \mathbb{N}_0.
\end{equation}
In analogy to the open-loop case, the steady-state equilibrium conditions of \eqref{eq:closed-loop} can be rewritten as
\begin{equation}\label{eq:Equilibrium_CL}
    G_xr = A_\pi G_xr + B_\pi r,\qquad \forall r\in\R^m,
\end{equation}
which, in turn, can be recast as
\begin{equation}\label{eq:open_loop_condition}
    (I-A_{\pi})G_x=B_\pi.
\end{equation}

\subsection{Linear Quadratic Regulator}
Let us consider the error coordinates $e_k=x_k-G_xr$, where $G_x,G_u$ satisfy \eqref{eq:G_xr_condition}. The LQR gain $K$ in \eqref{eq:LQR} can be computed in error coordinates by minimizing the $H_2$-norm of the closed-loop system
\begin{subequations}\label{eq:Input-Output}
    \begin{align}
        e_{k+1}&=(A+BK)e_k+w_k\\
        z_k&=\,\begin{bmatrix}
            Q^{1/2}~~~\\ R^{1/2}K
        \end{bmatrix} e_k,
    \end{align}
\end{subequations}
where $Q\succ0$ and $R\succ0$ are weighting matrices that penalize the state and input error $e_k$ and $e_k^u=u_k-G_ur$, respectively. 

As detailed in \cite{DD_Survey}, minimizing the $H_2$-norm of \eqref{eq:Input-Output} is equivalent to solving
\begin{subequations}\label{eq:LQR_modelBased}
    \begin{align}
        \min_{K,P}~~&\mathrm{trace}((Q+K^{\top}RK)P)\\
        \mbox{s.t.}~~ & (A+BK) P(A+BK)^\top-P+I \preceq 0,\\
        &P\succeq I,
    \end{align}
\end{subequations}
where $P \in \mathbb{R}^{n \times n}$ is the controllability Gramian of \eqref{eq:Input-Output}. Note that the LQR design problem \eqref{eq:LQR_modelBased} is independent of the equilibrium point that is being stabilized, thus shaping how insights into the equilibrium manifold inform direct data-driven control strategies.

\section{Indirect Data-Driven Control with insights on equilibrium manifolds}
A first approach that can be used to design the LQR gain $K$ in \eqref{eq:LQR} is to $i)$ estimate $(A,B)$ from the available data $\mathcal{D}_{T}$ in \eqref{eq:data}, and then $ii)$ solve \eqref{eq:LQR_modelBased} either under the certainty equivalence principle or by introducing a regularization to account for modeling errors as in \cite{regDD_2025}.    

To estimate the matrices, consider $X_0$ and $U_0$ in \eqref{eq:data_matrices0} and 
\begin{equation}\label{eq:data_W0}
    W_0 \coloneqq \begin{bmatrix} w_0 & w_1 & \cdots & w_{T-1} \end{bmatrix}.
\end{equation}
By constructing the state measurement matrix
\begin{equation}\label{eq:data_matrices1}
    X_1 \coloneqq \begin{bmatrix} x_1 & x_2 & \cdots & x_T \end{bmatrix}, 
\end{equation}
it is straightforward to see that 
\begin{equation}\label{eq:data_dynamics}
      X_1 = AX_0 + BU_0 + W_0.
\end{equation}
holds. Hence, the estimation of $A$ and $B$ can be carried out by solving the least squares problem
\begin{equation}\label{eq:iDD_classic}
\min_{A, B} \quad\frac{1}{2}\left\| X_1 - \begin{bmatrix}
       A & B
   \end{bmatrix}  \begin{bmatrix}
       X_0 \\ U_0
   \end{bmatrix}\right\|_F^2,
\end{equation}
resulting in
\begin{equation}\label{eq:Ind_LS}
    \begin{bmatrix}
       \hat{A}_{\mathrm{LS}} & \hat{B}_{\mathrm{LS}}
   \end{bmatrix} = X_1\begin{bmatrix}
       X_0 \\ U_0
   \end{bmatrix}^\dagger.
\end{equation}
Through standard arguments (see, e.g.,~\cite[App. II.2]{ljung1999}) and under Assumption~\ref{assump:process_noise}, these estimates can be proven to satisfy
\begin{subequations}\label{eq:bias_LS}
\begin{equation}
     \beta_{\mathrm{LS}}=\mathbb{E}\!\left[\mbox{vec}\!\left(W_0\begin{bmatrix}
        X_0 \\ U_0
    \end{bmatrix}^\dagger\right)\!\right]\!,
\end{equation}
where 
\begin{equation}
\beta_{\mathrm{LS}}=\mathbb{E}\!\left[\mbox{vec}\!\left(\begin{bmatrix}
        \hat{A}_{\mathrm{LS}} \!&\! \hat{B}_{\mathrm{LS}}
    \end{bmatrix}\right)\right]-\mbox{vec}\!\left(\begin{bmatrix}
        A \!&\! B
    \end{bmatrix}\right),
\end{equation}    
\end{subequations}
hence being \emph{unbiased} when $W_0$ is uncorrelated with $X_0$ and $U_0$. Moreover, it can be proven that
\begin{subequations}\label{eq:cov_LS}
\begin{equation}
    \mbox{cov}\!\left(\mbox{vec}\!\left(\begin{bmatrix}
        \hat{A}_{\mathrm{LS}} \!&\! \hat{B}_{\mathrm{LS}}
    \end{bmatrix}\right)\right)=\mathbb{E}\left[\mbox{vec}\left(\varepsilon_{\mathrm{LS}}\right)\mbox{vec}\left(\varepsilon_{\mathrm{LS}}\right)^{\!\top} \right],
\end{equation}
where
\begin{equation}
    \varepsilon_{\mathrm{LS}}=W_0\begin{bmatrix}
        X_0\\
        U_0
    \end{bmatrix}^{\dagger}\!\!-\mathbb{E}\left[W_0\begin{bmatrix}
        X_0\\
        U_0
    \end{bmatrix}^{\dagger}\right].\label{eq:varepsilon_LS}
\end{equation}
\end{subequations}
However, whenever one is provided not only with the data $\mathcal{D}_T$ but also with a pair $(G_x,G_u)$ satisfying \eqref{eq:nullSpace}, this additional information is not accounted for in this standard least squares problem. Therefore, the estimates in \eqref{eq:iDD_classic} could be inconsistent with the equilibrium subspace of the system.

To overcome this potential inconsistency, we propose to incorporate \eqref{eq:Equilibrium_subspace_general_condition} into the least squares problem \eqref{eq:iDD_classic}, leading to the following constrained identification problem:
\begin{subequations}\label{eq:iDD_equilibrium}
    \begin{align}
\min_{A, B} ~~&\frac{1}{2}\left\| X_1 - \begin{bmatrix}
       A & B
   \end{bmatrix}  \begin{bmatrix}
       X_0 \\ U_0
   \end{bmatrix}\right\|_F^2\\
\text{s.t.}~~ & \begin{bmatrix}
       A & B
   \end{bmatrix} \begin{bmatrix}
       G_x \\ G_u
   \end{bmatrix} = G_x,
\end{align}
\end{subequations}
whose solution is
    \begin{equation}\label{eq:iDDequilibrium_estimates}
        \begin{bmatrix}
            \hat{A} \!&\! \hat{B}
        \end{bmatrix}=X_1\begin{bmatrix}
            X_0\\
            U_0
        \end{bmatrix}^\dagger\left(I-\begin{bmatrix}
            G_x\\
            G_u
        \end{bmatrix}M\right)+G_x M,
    \end{equation}
    with 
    \begin{equation}\label{eq:aid_matrix}
    \begin{aligned}
        &M\!=\!\left(\!N\!\begin{bmatrix}
            G_x\\
            G_u
        \end{bmatrix}\right)^{\!\!-1}\!N \mbox{~~and~~}N\!=\!\begin{bmatrix}
            G_x\\
            G_u
        \end{bmatrix}^{\!\!\top}\!\!\left(\begin{bmatrix}
            X_0\\
            U_0
        \end{bmatrix}\begin{bmatrix}
            X_0\\
            U_0
        \end{bmatrix}^{\!\!\top}\right)^{\!\!-1}\!\!\!,
        \end{aligned}
    \end{equation}
as formalized, together with the solution's statistical properties, in the following theorem.
\begin{theorem}
   The solution of \eqref{eq:iDD_equilibrium} is \eqref{eq:iDDequilibrium_estimates}     
    and satisfies
    \begin{subequations}
        \begin{align}
           & \beta=\beta_{\mathrm{LS}}-\mathbb{E}\!\left[W_0\begin{bmatrix}
            X_0\\
            U_0 \end{bmatrix}^\dagger\!\begin{bmatrix}
            G_x\\
            G_u
        \end{bmatrix}M\!\right],\label{eq:bias_equilibrium}\\
        &\mbox{cov}\!\left(\mbox{vec}\left(\begin{bmatrix}
            \hat{A} \!&\! \hat{B}
        \end{bmatrix}\right)\right)=\mathbb{E}\left[\mbox{vec}\left(\varepsilon\right)\mbox{vec}\left(\varepsilon\right)^{\!\top} \right],\label{eq:covariance_iDD_equilibrium1}
        \end{align}
    with $\beta_{\mathrm{LS}}$ given in \eqref{eq:bias_LS},
    \begin{align}
        & \beta= \mathbb{E}\!\left[\mbox{vec}\!\left(\begin{bmatrix}
            \hat{A} \!&\! \hat{B}
        \end{bmatrix}\right)\!\right]-\mbox{vec}\!\left(\begin{bmatrix}
            A \!&\! B
        \end{bmatrix}\right),\\
        & \varepsilon\!=\!\varepsilon_{\mathrm{LS}}\!-\!\!\left(W_0\!\begin{bmatrix}
            X_0\\
            U_0
        \end{bmatrix}^{\!\dagger\!}\!\begin{bmatrix}
            G_x\\
            G_u
        \end{bmatrix}\!M\!\!-\!\!\mathbb{E}\!\left[W_0\!\begin{bmatrix}
            X_0\\
            U_0
        \end{bmatrix}^{\!\dagger\!}\!\begin{bmatrix}
            G_x\\
            G_u
        \end{bmatrix}\!M\!\right]\!\right)\!,\label{eq:covariance_iDD_equilibrium2}
    \end{align}
        \end{subequations}
        and $\varepsilon_{\mathrm{LS}}$ defined in \eqref{eq:varepsilon_LS}.
\end{theorem}
\begin{proof}
    Since \eqref{eq:iDD_equilibrium} is a convex problem, its explicit solution can be retrieved through the associated Karush-Kuhn-Tucker (KKT) conditions. In particular, consider the Lagrangian associated with \eqref{eq:iDD_equilibrium}, namely
    \begin{equation*}
        \mathcal{L}\!=\!\frac{1}{2}\left\| X_1 \!-\! \begin{bmatrix}
       A \!&\! B
   \end{bmatrix}\!  \begin{bmatrix}
       X_0 \\ U_0
   \end{bmatrix}\right\|_F^2\!+\mathrm{tr}\!\left(\left[\!\begin{bmatrix}
       A \!&\! B
   \end{bmatrix} \!\begin{bmatrix}
       G_x \\ G_u
   \end{bmatrix} \!-\! G_x\!\right]\!\mu^{\top\!\!}\right),
    \end{equation*}
    where $\mu \in \mathbb{R}^{n \times m}$ are the Lagrangian multipliers associated with the equilibrium subspace constraint. The KKT conditions satisfied by the optimal solution of \eqref{eq:iDD_equilibrium} are:
    \begin{subequations}
        \begin{align}
            & -\left(X_1 \!-\! \begin{bmatrix}
       \hat{A} \!&\! \hat{B}
   \end{bmatrix}  \begin{bmatrix}
       X_0 \\ U_0
   \end{bmatrix}\right)\begin{bmatrix}
       X_0 \\ U_0
   \end{bmatrix}^{\!\top}+\mu^\star\begin{bmatrix}
       G_x \\ G_u
   \end{bmatrix}^{\!\top}\!=\!0,\label{eq:KKT1}\\
            & \begin{bmatrix}
       \hat{A} \!&\! \hat{B}
   \end{bmatrix}\begin{bmatrix}
       G_x \\ G_u
   \end{bmatrix}-G_x=0.\label{eq:KKT2}
        \end{align}
    \end{subequations}
    By expressing the estimate $\begin{bmatrix}
       \hat{A} \!&\! \hat{B}
   \end{bmatrix}$ as a function of $\mu^\star$ via \eqref{eq:KKT1} and replacing it into \eqref{eq:KKT2}, we get
   \begin{equation*}
       \mu^\star=\left(X_1\begin{bmatrix}
           X_0\\
           U_0
       \end{bmatrix}^\dagger\begin{bmatrix}
           G_x\\
           G_u
       \end{bmatrix}-G_x\right)\left(N\begin{bmatrix}
           G_x\\
           G_u
       \end{bmatrix}\right)^{\!\!-1},
   \end{equation*}
   where $N$ is defined in \eqref{eq:aid_matrix}. Further replacing this into \eqref{eq:KKT1} leads to \eqref{eq:iDDequilibrium_estimates}. We are thus left to prove the statistical properties of the estimator. To this end, let us exploit \eqref{eq:data_dynamics} to equivalently rewrite \eqref{eq:iDDequilibrium_estimates} as
   \begin{align}
       \nonumber \begin{bmatrix}
            \hat{A} \!&\! \hat{B}
        \end{bmatrix}&=\begin{bmatrix}
            A \!&\! B
        \end{bmatrix}-\left(\begin{bmatrix}
            A \!&\! B
        \end{bmatrix}
        \begin{bmatrix}
            G_x\\
            G_u
        \end{bmatrix}-G_x\right)M
        +\\
        \nonumber &\quad\quad+W_0\begin{bmatrix}
            X_0\\
            U_0
        \end{bmatrix}^\dagger\left(I-\begin{bmatrix}
            G_x\\
            G_u
        \end{bmatrix}M\right)=\\
        &=\begin{bmatrix}
            A \!&\! B
        \end{bmatrix}+W_0\begin{bmatrix}
            X_0\\
            U_0
        \end{bmatrix}^\dagger\left(I-\begin{bmatrix}
            G_x\\
            G_u
        \end{bmatrix}M\right).
   \end{align}
   By computing the mean and vectorizing, the result in \eqref{eq:bias_equilibrium} straightforwardly follows. By leveraging the previous expression and the additivity of the mean, it is easy to show that
   \begin{align}
       \nonumber &\begin{bmatrix}
            \hat{A} \!&\! \hat{B}
        \end{bmatrix}\!-\mathbb{E}\!\left[\begin{bmatrix}
            \hat{A} \!&\! \hat{B}
        \end{bmatrix}\right]\!=\!\varepsilon_{\mathrm{LS}}+\\
        &~~~~-\!\left(W_0\!\begin{bmatrix}
            X_0\\
            U_0
        \end{bmatrix}^{\!\dagger\!}\!\begin{bmatrix}
            G_x\\
            G_u
        \end{bmatrix}M\!-\!\mathbb{E}\!\left[W_0\!\begin{bmatrix}
            X_0\\
            U_0
        \end{bmatrix}^{\!\dagger\!}\!\begin{bmatrix}
            G_x\\
            G_u
        \end{bmatrix}M\!\right]\!\right),
   \end{align}
    which corresponds to $\varepsilon$ in \eqref{eq:covariance_iDD_equilibrium2}.
\end{proof}
Note that, once again from standard arguments (see \cite{ljung1999}), under Assumption~\ref{assump:process_noise} the estimate is \emph{unbiased} when $W_0$ is uncorrelated with $X_0$ and $U_0$. Nonetheless, as expected, when these signals are correlated the bias $\beta_{\mathrm{LS}}$ in \eqref{eq:bias_LS} is modified according to the information on the equilibrium manifold. The variance of the estimated parameters is instead always affected by the prior we have enforced.

\section{Direct Data-Driven Control}
Rather than using a two-step approach, direct methods employ a parametrization of the closed-loop dynamics to incorporate measurements into \eqref{eq:LQR_modelBased} without performing a preliminary identification step. Specifically, here we employ the approach in \cite{regDD_2025}, which exploits the covariance transform matrices
\begin{equation}
    \begin{array}{ll}
    \tilde X_1=\frac{1}{T}X_1\begin{bmatrix}
       X_0 \\ U_0
   \end{bmatrix}^\top, & \tilde W_0=\frac{1}{T}W_0\begin{bmatrix}
       X_0 \\ U_0
   \end{bmatrix}^\top,\\
   \tilde X_0=\frac{1}{T}X_0\begin{bmatrix}
       X_0 \\ U_0
   \end{bmatrix}^\top, & ~\tilde U_0=\frac{1}{T}U_0\begin{bmatrix}
       X_0 \\ U_0
   \end{bmatrix}^\top,
    \end{array}
\end{equation}
to build the data-driven closed-loop representation used for the direct design of the LQR in \eqref{eq:LQR}. Specifically, introducing the parameterization
\begin{equation}\label{eq:Covar}
    \begin{bmatrix}
        I\\K
    \end{bmatrix} = \begin{bmatrix}
    \tilde X_0 \\
    \tilde U_0\end{bmatrix}V,
\end{equation}
the closed-loop dynamics can be cast as
\begin{align}
    \nonumber x_{k+1}&=\begin{bmatrix}
        A & B
    \end{bmatrix}\begin{bmatrix}
        I\\
        K
    \end{bmatrix}x_k+w_k=\begin{bmatrix}
        A & B
    \end{bmatrix}\begin{bmatrix}
    \tilde X_0 \\
    \tilde U_0\end{bmatrix}Vx_k+w_k\\
    &=(\tilde X_1-\tilde W_0)Vx_k+w_k.
\end{align}
Given Assumption~\ref{assump:process_noise}, the property
\begin{equation}\label{eq:dd_closed_loop}
    A_\pi=A+BK=(\tilde X_1-\tilde W_0)V,
\end{equation}
 and $T$ sufficiently large so that $\tilde W_0\approx0$, we recast \eqref{eq:LQR_modelBased} as
\begin{subequations}\label{eq:LQR_dDD}
\begin{align} 
\min_{V,P} ~~&\mbox{trace}((Q+V^\top \tilde U_0^\top R\tilde U_0V)P)\\
\mbox{s.t.} ~~&  \tilde X_1 V P V^\top \tilde X_1^\top-P+I\preceq0,\\
&P\succeq I,\\
&\tilde X_0 V = I.
\end{align}
\end{subequations}
The solution $V$ can then be used to compute the LQR feedback gain $K=\tilde U_0V$ and the corresponding (approximate) closed-loop state matrix $A_\pi\approx \tilde X_1V$.\smallskip

\subsection{Exploiting the Equilibrium Subspace \& Opening the Loop}
Thanks to the parameterization \eqref{eq:Covar}, solving \eqref{eq:LQR_dDD} is equivalent to directly tackling the following problem:
\begin{subequations}
    \begin{align}
        \min_{K,P}~~&\mathrm{trace}((Q+K^{\top}RK)P)\\
        \mbox{s.t.}~~ & A_\pi PA_\pi^\top-P+I \preceq 0,\\
        &P\succeq I.
    \end{align}
\end{subequations}
Since $B_\pi$ does not appear in problem, the equilibrium subspace constraint \eqref{eq:open_loop_condition} does not add any new information to the optimization problem. Indeed, for any Schur matrix $A_\pi$, the equilibrium subspace is trivially satisfied by selecting
\begin{equation}\label{eq:B_pi}
B_\pi=(I-A_\pi)G_x,
\end{equation}
after solving \eqref{eq:LQR_dDD}. Therefore, unlike the indirect method, information about $(G_x,G_u)$ cannot be explicitly used to guide the design of the controller. Although this result is somewhat disappointing, the equilibrium subspace can instead be used to compute the open-loop matrices $(A,B)$ associated to the closed-loop matrix $A_\pi$.

Let $(\Gamma_x,\Gamma_u)$ be a pair of matrices that satisfy the null space requirement \eqref{eq:nullSpace} \emph{and} the identity
\begin{equation}\label{eq:Special_Gamma}
    \Gamma_u=K\Gamma_x+I.
\end{equation}
For these matrices, the closed-loop system satisfies $B_\pi=B(\Gamma_u -K\Gamma_x)=B$. It then follows from \eqref{eq:B_pi} that 
\begin{equation}\label{eq:OpenB}
B\approx (I-A_\pi)\Gamma_x.
\end{equation}
Since $A_\pi=A+BK$, it therefore possible to compute
\begin{equation}\label{eq:OpenA}
A\approx A_\pi-(I-A_\pi)\Gamma_x K.
\end{equation}
This result allows us to ``open'' the loop, extracting the matrices $(A,B)$ from the closed-loop matrix $A_\pi$. Note that this operation is performed without compressing the data for fitting purposes, but only by solving a set of equalities associated with the condition on the equilibrium subspace. The reconstructed open-loop model is therefore likely to be less accurate compared with identified ones. Nevertheless, it provides insight into the open-loop representation hidden in the direct control design procedure.

\section{Equilibrium Subspace Identification}\label{sec:EqID}
The matrices $(G_x,G_u)$ characterizing the equilibrium manifold of a system are often known a priori, even when the underlying dynamics of the controlled system are unknown. For example, mechanical systems usually feature trivial equilibrium points (\emph{position = reference} and \emph{velocity = zero}), irrespective of the system dynamics. 
Nonetheless, in cases where this information is not available a priori, $(\Gamma_x,\Gamma_u)$ can be estimated using static experiments. Given a stabilizing feedback gain $K$ obtained by solving \eqref{eq:LQR_modelBased}+\eqref{eq:iDD_classic} or \eqref{eq:LQR_dDD}, let
\begin{equation}
    R=\begin{bmatrix}r_1&r_2&\ldots&r_p\end{bmatrix},
\end{equation}
be a collection of $p\geq m$ references, with $\mbox{rank}(R)=m$. Given $p$ independent experiments using the control inputs
\begin{equation}
    u_{k|i}=Kx_{k|i}+r_i,~~~i=1,\ldots,p,
\end{equation}
the resulting closed-loop responses
\begin{equation}
    x_{k+1|i}=(A+BK) x_{k|i}+Br_i+w_{k|i},
\end{equation}
satisfy the input-to-state stable equilibrium conditions
\begin{subequations}\label{eq:Eq_experiments}
\begin{align}
        \bar x_i &=(A+BK) \bar x_i+Br_i,\\
        \bar u_i &= K\bar x_i+r_i,\label{eq:ubar}
\end{align}
\end{subequations}
for all $i=1,\ldots,p$. Given the equilibrium map $\bar x_i=\Gamma_x r_i$ and $\bar u_i=\Gamma_u r_i$, let
\begin{subequations}
\begin{align}
    \bar X&=\begin{bmatrix}
        \bar x_1&\bar x_2&\ldots&\bar x_p\end{bmatrix}
    ,\\ \bar U&=\begin{bmatrix}
        \bar u_1&\bar u_2&\ldots&\bar u_p\end{bmatrix}.
\end{align}
\end{subequations}
It then follows from $\bar X=\Gamma_xR$ and $\bar U=\Gamma_uR$ that the equilibrium subspace matrices can be obtained using
\begin{equation}\label{eq:EquilibriumID}
    \begin{bmatrix}
        \Gamma_x \\\Gamma_u 
    \end{bmatrix}=\begin{bmatrix}
        \bar X\\\bar U
    \end{bmatrix}R^\dagger.
\end{equation}
Moreover, \eqref{eq:Eq_experiments} is sufficient to show that $(\Gamma_x,\Gamma_u)$ will satisfy the requirements \eqref{eq:Equilibrium_CL} and \eqref{eq:Special_Gamma}. The only thing left to do is estimate the equilibrium points $\bar x_i$ from data, for all $i=1,\ldots,p$. Since $\mathbb{E}(w_{k|i})=0$ under Assumption~\ref{assump:process_noise}, 
\begin{equation}
    \bar x_i\approx\frac1{T_{eq}}\sum_{k=1}^{T_{eq}}x_{\tau+k|i},
\end{equation}
where $\tau>0$ must be sufficiently large for the transient response to be depleted and $T_{eq}>0$ is the averaging window.\smallskip

Replacing the equality \eqref{eq:EquilibriumID} within \eqref{eq:OpenB}-\eqref{eq:OpenA}, yields
\begin{subequations}
\begin{align}
        A&\approx\tilde X_1V-(I-\tilde X_1V)\bar XR^\dagger\tilde U_0V,\\
        B&\approx(I-\tilde X_1V)\bar XR^\dagger,
\end{align}
\end{subequations}
hence being functions of the collected data and the data-driven controller.

\section{Numerical Validation} \label{sec:examples}
\begin{figure}
    \centering
    \begin{tikzpicture}[
    mass/.style={draw, thick, minimum width=1.8cm, minimum height=1.8cm, font=\Large},
    spring/.style={thick, decorate, decoration={zigzag, pre length=0.5cm, post length=0.5cm, segment length=3mm, amplitude=2.5mm}},
    arrow/.style={-{Stealth[length=3mm, width=2mm]}, thick}
]

    \node[mass] (m1) at (0,0) {$m_1$};
    \node[mass] (m2) at (5.5,0) {$m_2$};

    \draw[spring] ([yshift=0.5cm]m1.east) -- node[above=4mm, font=\large] {$k$} ([yshift=0.5cm]m2.west);

    \draw[thick] ([yshift=-0.5cm]m1.east) -- ++(1.6,0) coordinate (c_start);
    \draw[thick] (c_start) ++(0.8, 0.3) -- ++(-0.8, 0) -- ++(0, -0.6) -- ++(0.8, 0);
    \draw[thick] ([yshift=-0.5cm]m2.west) -- ++(-1.6,0) coordinate (p_start);
    \draw[thick] (p_start) ++(0, 0.2) -- ++(0, -0.4);
    \node[font=\large] at (3, 0) {$c$};

    \draw[arrow] (-2.0, 0) -- node[above, font=\large] {$u_1$} (m1.west);

    \draw[arrow] (2.5, -1.2) -- node[below, font=\large] {$u_2$} (1.5, -1.2); 
    \draw[arrow] (3.5, -1.2) -- node[below, font=\large] {$u_2$} (4.5, -1.2); 

\end{tikzpicture}
    \caption{Two-mass mechanical system model used for our simulation study.}
    \label{fig:mechanical_model}
\end{figure}
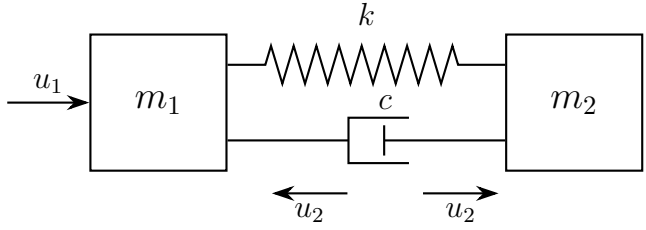
To compare the proposed equilibrium-informed approaches with a baseline, we perform a set of numerical experiments on the system shown in Figure \ref{fig:mechanical_model}, with $k=100$ and $c=2$. In our test, the true continuous-time state-space matrices of the system are
\begin{align*}
    A_c = \begin{bmatrix} 0 & 0 & 1 & 0 \\ 0 & 0 & 0 & 1 \\ -k~~ & k & -c~~ & c \\ k & -k~~ & c & -c~~ \end{bmatrix},\quad 
    &B_c = \begin{bmatrix} 0 & 0\!\! \\ 0 & 0\!\! \\ 1 & -1~~\!\! \\ 0 & 1 \!\!\end{bmatrix}\!,
\end{align*}
which, using a zero-order hold with sampling time $\tau_s=0.01$ to discretize the dynamics, led to
\begin{align*}
    &A = \begin{bmatrix} 0.9951 &   0.0049 &   0.0099   & 0.0001\\
    0.0049  &  0.9951  &  0.0001  &  0.0099\\
   -0.9770  &  0.9770  &  0.9755  &  0.0245\\
    0.9770 &  -0.9770  &  0.0245  &  0.9755\end{bmatrix},\\ 
    &B = \begin{bmatrix} 0.0000 &  -0.0000\\
    0.0000  &  0.0000\\
    0.0099 &  -0.0098\\
    0.0001 &   0.0098\!\!\end{bmatrix}\!.
\end{align*}
Note that these matrices do not inherit the sparsity of the continuous-time system. However, they do inherit its equilibrium manifold. After using the direct method\footnote{Although $K$ can be obtained using \emph{any} approach for the purpose of the indirect method, we use $K$ of the direct method for fairness of comparison.} to compute the LQR gain $K$, we performed the equilibrium subspace reconstruction procedure detailed in Section \ref{sec:EqID}, using $r_1=[10~~0]^\top$ and $r_2=[0~~10]^\top$ to obtain $\Gamma_x\text{, } \Gamma_u$. We then leveraged these matrices to estimate $(A,B)$ using \eqref{eq:iDDequilibrium_estimates}, and \eqref{eq:OpenB}-\eqref{eq:OpenA}. For the sake of comparison, we also estimated $(A,B)$ using \eqref{eq:Ind_LS}. All experiments were performed with white Gaussian process noise with $0.015$ standard deviation. The different approaches were compared  over 1000 Monte Carlo simulations, using the error metric
\begin{equation}\label{eq:error}
    \rm{error}=\left\|\mbox{vec}\!\left(\begin{bmatrix}
        \hat{A} \!&\! \hat{B}
    \end{bmatrix}\right)-\mbox{vec}\!\left(\begin{bmatrix}
        A \!&\! B
    \end{bmatrix}\right)\right\|_{2}.
\end{equation}
As summarized in Table \ref{tab:monte_carlo_results}, incorporating the equilibrium subspace into the indirect method yields a better approximation of the true model of the system compared to the classic indirect approach. As for the direct method, the proposed approach allows us to estimate the open-loop system. Unsurprisingly, the quality of the estimate obtained by ``opening'' the direct method is lower than that of the indirect methods. This is justified by the fact that the direct method prioritizes the control objective over fitting performance.

\begin{table}
\vspace{10pt}
    \centering
    \caption{Monte Carlo Results (1000 Trials): methods \emph{vs} statistics of the error \eqref{eq:error}.}
    \label{tab:monte_carlo_results}
    \begin{tabular}{lcccc}
        \toprule
        \textbf{Method} & \textbf{Average} & \textbf{Min} & \textbf{Max} & \textbf{Std. Dev.} \\
        \midrule
        \text{Indirect (Classic)} & 0.04026   & 0.01092   & 0.08766   & 0.01230    \\
        \text{Indirect (Equilibrium)}    & \textbf{0.03576}   & \textbf{0.00846}   & \textbf{0.08386}   & \textbf{0.01177}    \\
        \text{Direct }     & 0.06520   & 0.01211  &  0.74169   & 0.04424  \\        
        \bottomrule
    \end{tabular}
\end{table}

\section{Experimental Validation}
Experimental validation was conducted using the Quanser rotary flexible joint module in Figure \ref{fig:Sketch Module}. The system is similar to the one in Section \ref{sec:examples}, except that there is no input $u_2$ and the second degree of freedom is defined using error coordinates. 
\begin{figure}
    \centering
    \begin{tikzpicture}[>=Latex, thick]

    \def\thetaOne{30} 
    \def\thetaTwo{45}  
    
    \def\hubRadius{1.5}
    \def\armLength{4.5}
    \def\springAttachArm{2.5} 

    \draw[dashed, thin, gray] (-1,0) -- (4,0) node[right, text=black] {Reference ($0^\circ$)};

    \draw[->, blue!80!black, thick] (2,0) arc[start angle=0, end angle=\thetaOne, radius=2];
    \node[blue!80!black] at (\thetaOne/2:2.22) {$\theta_1$};

    \draw[->, red!80!black, thick] (\thetaOne:2.8) arc[start angle=\thetaOne, end angle=\thetaTwo, radius=2.8];
    \node[red!80!black] at ({(\thetaOne+\thetaTwo)/2-1}:3.1) {$\theta_2$};

    \begin{scope}[rotate=\thetaOne]
        \draw[draw=black, thick, rounded corners=2pt] (-0.5, -0.8) rectangle (\hubRadius, 0.8);
        
        \draw[dashed, thick, blue!80!black] (0,0) -- (3.5,0);

        \filldraw[orange!80!black, draw=black] (\hubRadius-0.2, 0.6) circle (0.1) coordinate (H1);
        \filldraw[orange!80!black, draw=black] (\hubRadius-0.2, -0.6) circle (0.1) coordinate (H2);
        \node[black!80!black, align=center] at (0.5, -1.1) {Hub};
    \end{scope}

    \begin{scope}[rotate=\thetaTwo]
        \draw[fill=black!80, draw=black, thick, rounded corners=1pt] (0, -0.1) rectangle (\armLength, 0.1);
        \filldraw[orange!80!black, draw=black] (\springAttachArm, 0) circle (0.1) coordinate (A1);
        \node[black] at (\armLength/1.3, 1) {Flexible Arm};
    \end{scope}

    \draw[decorate, decoration={coil, aspect=0.4, segment length=2mm, amplitude=1.5mm}, thick, gray!80] (H1) -- (A1);
    \draw[decorate, decoration={coil, aspect=0.4, segment length=2mm, amplitude=1.5mm}, thick, gray!80] (H2) -- (A1);

    \filldraw[black] (0,0) circle (0.2);
    \filldraw[gray!50] (0,0) circle (0.1);

\end{tikzpicture} 
    \caption{Sketch of the Rotary Flexible Joint Module}
    \label{fig:Sketch Module}
\end{figure}
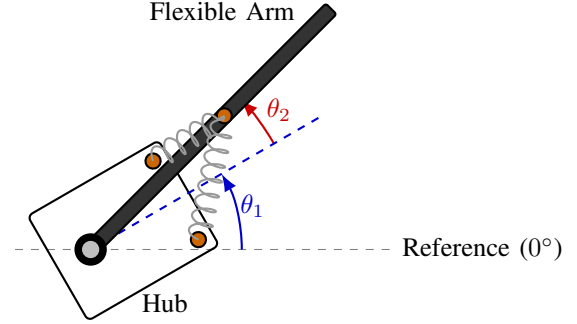
The equilibrium subspace for this system is trivial: Given an input torque $u=0$, the angular position $\theta_1$ can take arbitrary values, whereas the presence of the spring guarantees $\theta_2=0$. As for the velocities, they are necessarily zero at equilibrium. Based on these considerations, a suitable representation of the equilibrium subspace of the system is 
\begin{equation*}
         G_x = \begin{bmatrix} 1 & 0 & 0 & 0 \end{bmatrix}^\top\!\!\!\!,\quad G_u = \begin{bmatrix} 0\end{bmatrix}. 
\end{equation*}
These matrices satisfy \eqref{eq:nullSpace} and can therefore be used to inform the identification of the system. The controllers are implemented at $500\,\mathrm{Hz}$ using $Q=\mathrm{diag}([1 ~~1.5~~ 0.01~~0])$ and $R=0.01$. The initial data was collected over $20\,\mathrm{s}$.

\begin{figure}
    \centering
    \includegraphics[width=0.8\linewidth]{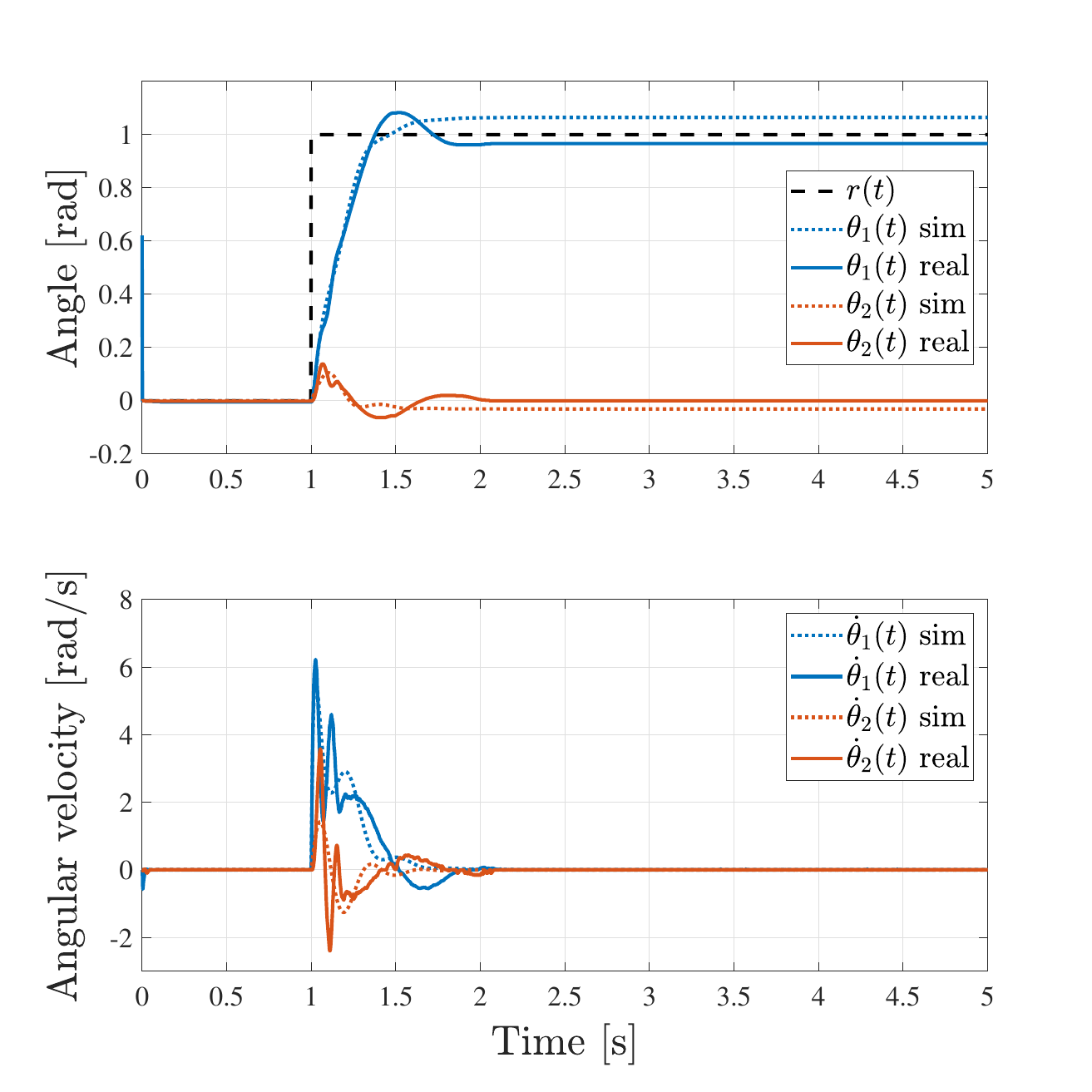}
    \vspace{-12pt}
    \caption{Closed-loop response obtained with the classic indirect method. The model features a significant steady-state mismatch between the measured data (real) and the identified model (sim).}
    \label{fig:Classic}
    \includegraphics[width=0.8\linewidth]{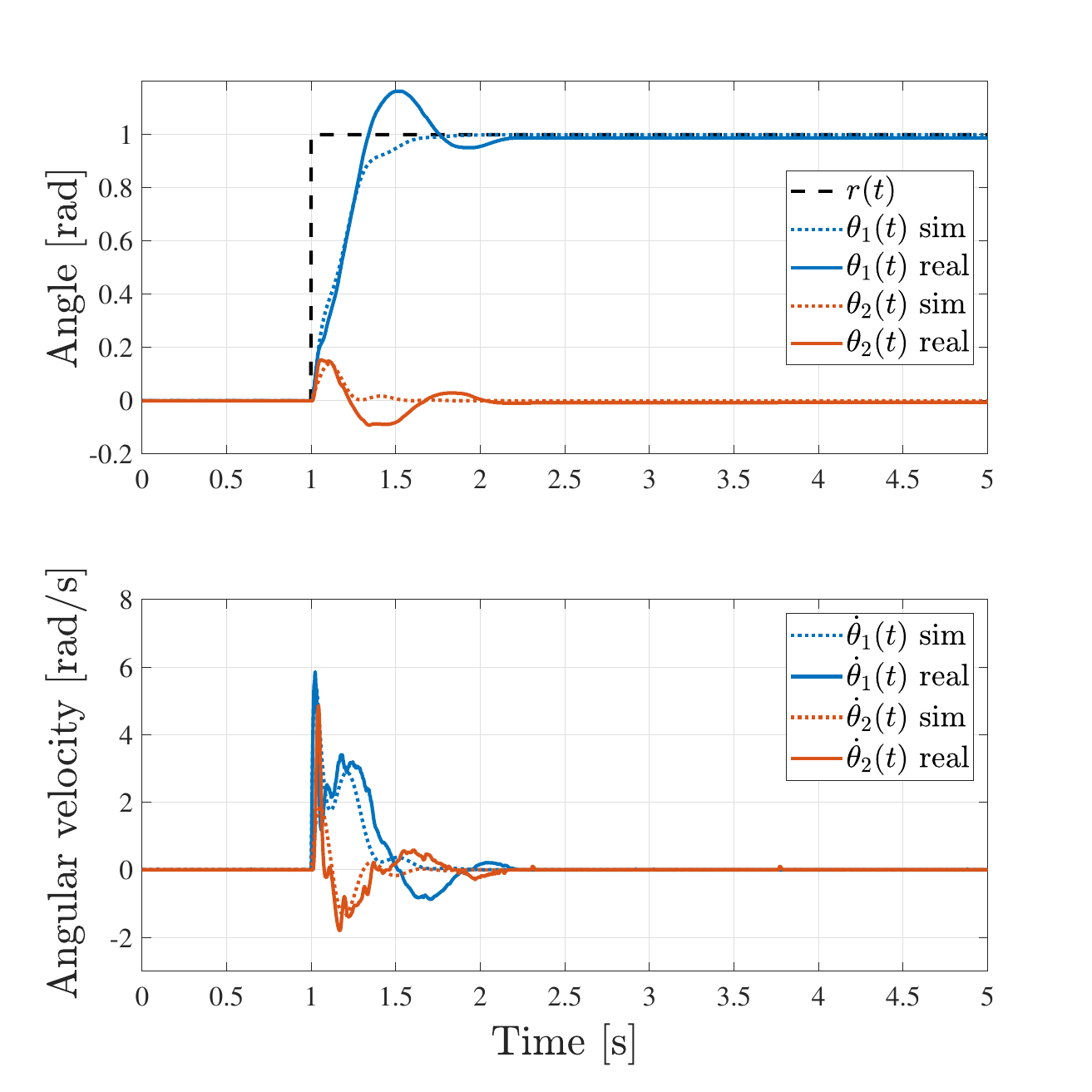}
    \vspace{-12pt}
    \caption{Closed-loop response obtained with the direct method. The steady-state error is negligible.}
    \label{fig:Direct}
    \includegraphics[width=0.8\linewidth]{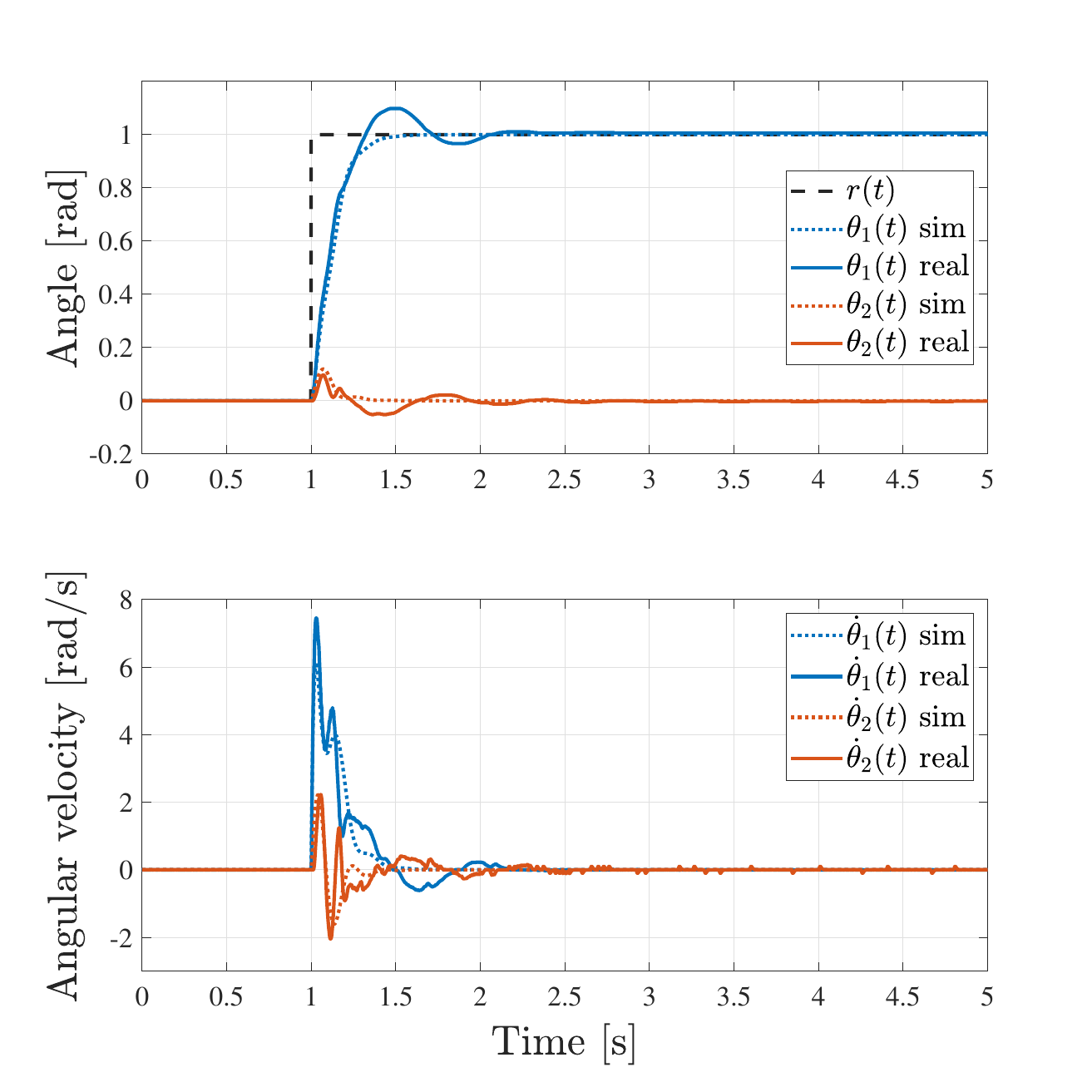}
    \vspace{-12pt}
    \caption{Closed-loop response obtained with the indirect method, augmented with equilibrium subspace priors. The steady-state error is negligible and the overshoot is smaller than that obtained with direct method.}
    \label{fig:Equilibrium}
\end{figure}

Figures \ref{fig:Classic} and \ref{fig:Direct} report the closed-loop responses obtained using the classic indirect method \eqref{eq:LQR_modelBased} + \eqref{eq:iDD_classic} and the covariance-based direct method in \eqref{eq:LQR_dDD}, respectively. The former results in a significant model mismatch at equilibrium. The latter achieves a much better response, with virtually no steady-state error even in the absence of an integrator. 
Figure \ref{fig:Equilibrium} reports the closed-loop response obtained using the proposed indirect method, i.e., \eqref{eq:LQR_modelBased} + \eqref{eq:iDD_equilibrium}. Here, we see that incorporating the equilibrium subspace prior significantly improves closed-loop performance of the indirect method, making its closed-loop response comparable to the direct method at the price of requiring additional information on the system. It is worth remarking that the transient achieved using the proposed indirect method is faster than the one attained with the direct one, at the cost of a more aggressive control action.

\section{Conclusion}
This paper details how knowledge of the equilibrium manifold can benefit both direct and indirect data-driven methods. For indirect methods, we show how this prior can be directly embedded into the system identification problem. For direct methods, we show how it can be used to recover an estimate of the open-loop dynamics starting from the closed-loop one. Our simulation and experimental results provide evidence of the potential advantages and drawbacks of the proposed procedure. 

\bibliographystyle{IEEEtran}
\bibliography{references}

\end{document}